\documentclass[12pt,a4paper]{article}
\usepackage{verbatim}
\usepackage[utf8]{inputenc} 
\usepackage[english]{babel} 
\usepackage{anysize} 
\usepackage{indentfirst}
\usepackage{amsmath}
\usepackage{amsthm} 
\usepackage{amsfonts}
\usepackage{amssymb}
\usepackage{graphicx}
\usepackage{enumitem}

\graphicspath{{figuras/}}
\usepackage{subfig} 
\usepackage{xcolor}

\newtheorem{theorem}{Theorem}[section]
\newtheorem{proposition}[theorem]{Proposition}

\newtheorem{definition}[theorem]{Definition}

\newtheorem{remark}[theorem]{Remark}

\newcommand{\R}{\mathbb{R}}

\newcommand{\IR}{\mathbb{I}(\mathbb{R})}
\newcommand{\IRn}{\mathbb{I}(\mathbb{R}_{-})}
\newcommand{\IRp}{\mathbb{I}(\mathbb{R}_{+})}

\newcommand{\prt}[1]{\langle #1 \rangle}

\title{On Monotonicities of Interval Valued Functions}

\author{Ana Shirley Monteiro \and Regivan Santiago \and Martin Papco \and Radko Mesiar \and Humberto Bustince}

\date{}

\begin{document}
	
\maketitle 


\section{Introduction}

In this paper, we develop the notions of weak/directional monotonicity (developed by Sesma-Sara et al. \cite{MikelIVPreAgg2019} in terms of the \emph{Kulisch-Miranker order}) and the notion of $ G $-monotonicity (introduced by Santiago et al. \cite{SantiAbsHomo2021} for $ [0,\!1]$)
for intervals endowed with admissible orders.


\section{\label{sec:Prelim} Preliminaries}

The following notation will be used: $\mathbb{I}(\mathbb{R})$  denotes the set of all closed intervals with real numbers, uppercase letters denote its elements, $\IRn$ denotes the subset  consisting of all closed intervals with negative real numbers, $\IRp$ denotes the subset of closed intervals with positive real numbers. The set $\mathbb{I}([0,\!1])$  contains all closed subintervals of $[0,\!1]$. Further, we use uppercase letters to denote any interval variable $X$ and use bold fonts to denote degenerated intervals; $\vec{0}$ and $\vec{1}$ to denote the vectors $(0, \ldots, 0)$ and $(1, \ldots, 1)$, respectively; $\vec{\boldsymbol{0}}$ and $\vec{\boldsymbol{1}}$, in bold, to denote the vectors $(\boldsymbol{0}, \ldots, \boldsymbol{0})$ and $(\boldsymbol{1}, \ldots, \boldsymbol{1})$, respectively, where $\boldsymbol{0} = [0,0]$ and $\boldsymbol{1} = [1,1]$. The prefix ``\emph{IV}'' will mean ``\emph{interval-valued}'' in general.

\medskip

\begin{definition}[Operations over intervals]
\label{def_intoperations}
Let $X,Y\in \IR$ and $\alpha\in\R$.
\begin{enumerate}[leftmargin=*,label={\rm($\mathbf{OI}$-\arabic*)}]
		
\item The \emph{sum of $X$ and $Y$} is defined by $X + Y = [\underline{X} + \underline{Y}, \overline{X} + \overline{Y}]$.
		
\item The \emph{opposite of $X$} is defined by\  $-X = [-\overline{X}, - \underline{X}]$.
		
\item The \emph{difference of $X$ and $Y$} is defined by $X - Y = [\underline{X} - \overline{Y},  \overline{X} - \underline{Y}]$.
		
\item The \emph{product of $X,Y\in\IRp$} is defined by $X \cdot Y = [\underline{X} \cdot \underline{Y}, \overline{X} \cdot \overline{Y}]$.
		
\item The \emph{$\alpha$-multiple of $X$} is defined by $\alpha\cdot X=[\alpha \cdot \underline{X}, \alpha \cdot \overline{X}]$ if $\alpha \geq 0$ and  $\alpha\cdot X=[\alpha \cdot \overline{X}, \alpha \cdot \underline{X}]$ if $\alpha <0$.
\end{enumerate}
\end{definition}


\section{Monotonicities of IV-Functions}

In the case of intervals, considering the \textit{Kulisch-Miranker order} (${\preceq}_{KM}$) on $\mathbb{I}([0,\!1])$, Sesma-Sara et al. \cite{MikelIVPreAgg2019} provided the definitions of standard monotonicity, weak monotonicity, directional monotonicity and pre-aggregation functions to the interval context as follows. 

\begin{definition}
	\label{def_intmonotonicityKM}
	A function $F: \mathbb{I}([0,\!1])^n \to \mathbb{I}([0,\!1])$ is said to be increasing (resp. decreasing) if for all $X, Y \in \mathbb{I}([0,\!1])$ such that $X \  {\preceq}_{KM^n} \  Y$ it holds that $F(X) \ {\preceq}_{KM}  \ F(Y)$ (resp. $F(X) \  {\succeq}_{KM}  \ F(Y)$).
\end{definition}

\begin{definition} 
\label{def_intrmonotonicityKM}
Let $\overrightarrow{V} = \prt{(a_1,b_1), \ldots, (a_n,b_n)} \in (\mathbb{R}^2)^n$ such that $(a_i, b_i) \not= \vec{0}$ for some $i \in \{1, \ldots,n\}$.  A function $F: \mathbb{I}([0,\!1])^n \to \mathbb{I}([0,\!1])$ is said to be $\overrightarrow{V}$-increasing (resp. $\overrightarrow{V}$-decreasing) if for all $X\in \mathbb{I}([0,\!1])$ and $c > 0$ such that $X + c\overrightarrow{V} \in \mathbb{I}([0,\!1])^n$, it holds that $F(X) \ {\preceq}_{KM} \  F(X+c\overrightarrow{V})$ (resp. $F(X) \  {\succeq}_{KM} \  F(X+c\overrightarrow{V})$). In the case that $F$ is simultaneously $\overrightarrow{V}$-increasing and $\overrightarrow{V}$-decreasing, $F$ is said to be $\overrightarrow{V}$-constant.
\end{definition}

\begin{definition} 
	\label{def_intweakmonotonicityKM}
	Let $(a,b) \in \mathbb{R}^2$ with  $(a, b) \not= \vec{0}$ and $\overrightarrow{V} = \prt{(a, b), \ldots, (a,b)}$. A function $F: \mathbb{I}([0,\!1])^n \to \mathbb{I}([0,\!1])$
	is said to be $(a,b)$-weakly increasing (resp. $(a,b)$-weakly decreasing) if for all $x \in \mathbb{I}([0,\!1])^n$ and $c > 0$ such that $x + c\overrightarrow{V} \in \mathbb{I}([0,\!1])^n$, it holds that $F(X) \ {\preceq}_{KM} \  F(X+c\overrightarrow{V})$ (resp. $F(X) \ {\succeq}_{KM} \  F(X+c\overrightarrow{V})$).  
\end{definition}

The previous definitions use vectors of $(\mathbb{R}^2)^n$ and the Kulisch-Miranker order on $\mathbb{I}([0,\!1])$. In what follows we provide definitions which are based on admissible orders.

\color{black}

\begin{definition} 
	\label{def_intweakmonotonicity}
	An IV-function $F : \mathbb{I}([0,\!1])^n \to \mathbb{I}([0,\!1])$ is said to be \emph{weakly increasing} if 
	\begin{center}
		$F(X_1 + C, \ldots, X_n + C) \succeq F(X_1, \ldots, X_n)$,
	\end{center} 
	for all $X_1, \ldots, X_n, C \in \mathbb{I}([0,\!1])$ such that  $(X_1 + C, \ldots, X_n + C) \in \mathbb{I}([0,\!1])^n$. Dually we define a \emph{weakly decreasing} IV-function.
\end{definition}

In what follows we introduce the notions of $\overrightarrow{V}$\!-increasing and $\overrightarrow{V}$\!-decreasing  IV-functions. 

\vspace{0,1cm}

\begin{definition}
	\label{def_intrmonotonicity}
	Given $\overrightarrow{V} = (\boldsymbol{V_1}, \ldots, \boldsymbol{V_n}) \in \mathbb{I}(\mathbb{R})^n \setminus \{\vec{\boldsymbol{0}}\}$ such that each $\boldsymbol{V_i}$ is a degenerated interval for $i = 1, \ldots, n$,  an IV-function $F : \mathbb{I}([0,\!1])^n \to \mathbb{I}([0,\!1])$ is said to be  \emph{$\overrightarrow{V}$\!-increasing}, if 
	\begin{center}
		$F(X_1 + k \cdot \boldsymbol{V_1}, \ldots, X_n + k \cdot \boldsymbol{V_n}) \succeq F(X_1, \ldots, X_n)$, 
	\end{center}
	for all  $0 < k \in \mathbb{R}$ and $(X_1,  \ldots, X_n) \in \mathbb{I}([0,\!1])^n$ such that $(X_1 + k \cdot \boldsymbol{V_1}, \ldots, X_n + k \cdot \boldsymbol{V_n}) \in \mathbb{I}([0,\!1])^n$. Dually we define  \emph{$\overrightarrow{V}$\!-increasing} IV-function. In the case that $F$ is simultaneously $\overrightarrow{V}$\!-increasing and $\overrightarrow{V}$\!-decreasing, $F$ is said to be  a $\overrightarrow{V}$\!-constant IV-function.
\end{definition}

\begin{remark}
\label{remark_intweakdirmonotonicity} The weak increase (resp. the weak decrease) implies the $\vec{\boldsymbol{1}}$-increase (resp. the $\vec{\boldsymbol{1}}$-decrease)  of IV-functions.
\end{remark}

\begin{proposition}
\label{prop_directmonotImp}
Every IV-implication $I: \mathbb{I}([0,\!1])^2 \to \mathbb{I}([0,\!1])$ is a  $(\boldsymbol{-1}, \boldsymbol{1})$-increasing IV-function.
\end{proposition}

\begin{proof}
	Indeed, if $I: \mathbb{I}([0,\!1])^2 \to \mathbb{I}([0,\!1])$ is an IV-implication, then for every  $0 < k \in \mathbb{R}$ and $(X, Y) \in \mathbb{I}([0,\!1])^2$ such that $(X + k \cdot (\boldsymbol{-1}), \boldsymbol{y} + k \cdot \boldsymbol{1}) \in \mathbb{I}([0,\!1])^2$ we have: 
	\begin{center}
		$I(X + k \cdot (\boldsymbol{-1}), Y + k \cdot \boldsymbol{1}) = I(X - [k,k], Y + [k,k]) \succeq I(X, Y + [k,k]) \succeq I(X, Y)$,
	\end{center}
	since $ I $ is an IV-function  decreasing  in the first component and increasing in the second component.
\end{proof}

\section{$G\text{-}$ weak monotonicity}

Let us recall the notion of $g$-weakly increasing function on $U$ provided by Santiago et al. in \cite{SantiAbsHomo2021}:

\begin{definition}
\label{def_gweakmonfunction}
Let $g: [0,\!1]^2 \to [0,\!1]$ be a function such that $ g(x,y) \geq y $, for all $ x,y \in [0,\!1]$. A partial function $f: [0,\!1]^n \to [0,\!1]$ is said to be $g$-weakly increasing if $f(g(\lambda,x_1), \ldots, g(\lambda,x_n)) \geq f(x_1, \ldots, x_n)$, for all $(x_1, \ldots, x_n) \in [0,\!1]^n$ and $\lambda \in (0,\!1]$. Dually we define a $g$-weakly decreasing function.
\end{definition}

Next we provide the notion of $ G $-weak monotonicity to the interval context considering on $\mathbb{I}([0,\!1])$ admissible orders.

\begin{definition} 
\label{def_Gweakmonotonicity}
Let $G : \mathbb{I}([0,\!1])^2 \to \mathbb{I}([0,\!1])$ be an IV-function such that $G(X,Y) \succeq Y $ for all $X,Y \in \mathbb{I}([0,\!1])$. An IV-function $F: \mathbb{I}([0,\!1])^n \to \mathbb{I}([0,\!1])$ is said to be  \emph{$G$-weakly increasing}, if \begin{center}
$F\left(G(\Lambda,X_1), \ldots, G(\Lambda,X_n)\right) \succeq F(X_1, \ldots, X_n) $,
\end{center} 
for all $\Lambda \in \mathbb{I}([0,\!1])$ and $(X_1, \ldots, X_n) \in \mathbb{I}([0,\!1])^n$. Dually we define a \emph{$G$-weakly decreasing} IV-function.
\end{definition}

\begin{proposition} 
\label{prop_intGweakmonotonicity}
Every increasing (resp. decreasing) IV-function $F : \mathbb{I}([0,\!1])^n \to \mathbb{I}([0,\!1])$ is a  $G$-weakly increasing (resp. $G$-weakly decreasing) IV-function for every IV-function $G : \mathbb{I}([0,\!1])^2 \to \mathbb{I}([0,\!1])$ such that $G(X,Y) \succeq Y$.
\end{proposition}

\begin{proof}
Straightforward.
\end{proof}



\end{document}